\documentclass[a4paper,fleqn]{cas-dc}

\usepackage[square, comma, sort&compress, numbers]{natbib}
\usepackage{hyperref}
\def\tsc#1{\csdef{#1}{\textsc{\lowercase{#1}}\xspace}}
\tsc{WGM}
\tsc{QE}
 
\newtheorem{theorem}{Theorem}
\newtheorem{definition}{Definition}
\newtheorem{corollary}{Corollary}

\newdefinition{remark}{Remark}
\newdefinition{proposition}{Proposition}
\newproof{proof}{Proof}
\newproof{pot}{Proof of Theorem \ref{thm}}

\allowdisplaybreaks[4]
\usepackage{amsmath}
\usepackage{cuted} 
\begin{document}
\let\WriteBookmarks\relax
\def\floatpagepagefraction{1}
\def\textpagefraction{.001}

\shorttitle{Evolutionary dynamics with frequency-dependent returns}    

\shortauthors{Wang et~al.}  

\title [mode = title]{Evolutionary dynamics in public goods games with general  frequency-dependent returns}  



%

\author[1]{Qiushuang Wang}
\author[1]{Yuyuan Liu}
\author[1]{Xiaojie Chen}
\cormark[1]
\ead{xiaojiechen@uestc.edu.cn}

\author[2]{Attila Szolnoki}

\affiliation[1]{organization={School of Mathematical Sciences, University of Electronic Science and Technology of China},
	city={Chengdu},
	postcode={611731}, 
	country={People's Republic of China}}

\affiliation[2]{organization={Institute of Technical Physics and Materials Science, Centre for Energy Research},
	addressline={P.O. Box 49}, 
	city={Budapest},
	postcode={H-1525}, 
	country={Hungary}}

\cortext[cor1]{Corresponding author}


\begin{abstract}
	The public goods game serves as a significant paradigm for investigating the emergence and maintenance of cooperation in conflicting situations. In the traditional public goods game, the multiplication factor characterizing the synergy effect of common efforts is typically assumed to be constant. In real-world scenarios, however, investment returns are often dynamic and vary with the strategic composition of the interaction group. To date, the evolutionary dynamics of the public goods game with such frequency-dependent returns have remained not fully understood. In this work, we introduce a general frequency-dependent multiplication factor that depends on the strategy composition within the game group. Through theoretical analysis, we derive the mathematical conditions under which cooperation is favored. Our results show that whether cooperation has an evolutionary advantage over defection depends on the investment return rate in the full-contribution state of the game group, irrespective of the return rates in other states. An increase in this rate leads to a higher abundance of cooperators. Furthermore, we introduce a general frequency-dependent multiplication factor into the public goods game with peer punishment, and systematically explore its effects on the cooperation dilemma and the second-order free-rider problem. Our results highlight that the abundance of cooperators or punishers is governed solely by the investment return values in the full-cooperation and full-punishment compositions of the group. A higher return rate in the full-cooperation state facilitates the promotion of cooperation, whereas a higher return rate in the full-punishment state favors the emergence of punishment. Our theoretical findings are verified by individual-based simulations.
\end{abstract}
%


\maketitle

\section{Introduction}\label{Xsec1-1}\label{Xsec1-1}
Cooperation is a widespread phenomenon both in  nature and in human society. Understanding how cooperation emerges and evolves among self-interested individuals is a central question of great scientific importance~\cite{FehrNature2003}. Evolutionary game theory provides a powerful framework for studying the evolution of cooperation~\cite{AllenNature2017,HofbauerCambridgeUP1998,NowakHarvardUP2006,NowakScience2006,LeePD2025,OhtsukiNature2006,PalNC2022,TanimotoEESCS2018,PeiPD2024,LuoPD2021}. In real world settings, interactions often involve multiple individuals, and the public goods game has become a popular paradigm for studying cooperation in such multi-player interactions~\cite{ChenNJP2015,HauertJTB2002,LiuMMMAS2019,MilinskiPNAS2006,SantosNature2008,SzolnokiJTB2013}. In the classical public goods game, individuals choose between two options whether to cooperate or defect in a common venture. Cooperators bear personal costs to contribute to the public pool, while defectors contribute nothing. Then, all contributions are multiplied by a multiplication factor that governs the return of investments, whereupon the resultant total is distributed equally among all participants. As a result, defectors obtain higher payoffs than cooperators, being free-riders. Although cooperation enhances the overall payoff of the group, defection is the optimal strategy for rational individuals. This leads to a social dilemma, implying that in the context of evolutionary game theory, without additional mechanisms, natural selection disfavors the evolution of cooperation~\cite{HardinScience1968,TkadlecPNAS2023}.

In recent years, considerable research efforts have been devoted to addressing the problem of cooperation in the public goods game~\cite{ChenJRSI2015,GurerkScience2006,HauertJTB2010,Liuelife2023,RandNC2011,SasakiPRSB2013,WangNC2024}. A prominent research concern is the consideration of strategy-dependent feedback, which accounts for bidirectional influences between individual strategies and some factors such as environmental resources~\cite{PercBiosystems2009,ShaoEPL2019,wangqcsf25}, group size~\cite{LeeChaosSF2023}, or population state~\cite{ChenPloSone2012,HauertJRSI2024,HauertJTB2006,ShiPhysicaA2012,WangChaos2025}. Indeed, such feedback is also observed between the multiplication factor and the strategy pattern of the group~\cite{BejanJAP2017,StrausJAE2022}. For example, cell growth is regulated by growth factors, where the regulatory effect exhibits a nonlinear functional relationship with the concentration of growth factors~\cite{FrankBD2013,ArchettiJTB2016}.
Production cost is subject to economies of scale, where the marginal cost decreases with output. Thus, the total cost exhibits a nonlinear function of the quantity produced~\cite{BejanJAP2017}. Recently, Hauert and McAvoy~\cite{HauertJRSI2024}\break considered a feedback between the multiplication factor and group strategy choices by proposing that the multiplication factor depends linearly or nonlinearly on the number of contributors in the group. They found that a linear feedback leads to richer evolutionary dynamics and breaks the dominance of defection observed in the traditional public goods game.

An alternative way to foster cooperation is the usage of incentives in the traditional public goods game, including punishment~\cite{SongJRSI2026,HanAB2016,HauertScience2007,SantosNature2008,WangNC2024,ZhangChaos2025}, reward~\cite{HauertJTB2010,LuPRE2020,SigmundPNAS2001},  exclusion~\cite{LiuChaos2018,LiuJRSI2021,SasakiPRSB2013,QuanChaos2019}, and so on.  Notably, peer punishment has been proven to be an effective way of promoting cooperation in multi-player interactions~\cite{BoydPNAS2003,FehrNature2002,FowlerPNAS2005,GachterScience2008,RandNC2011,SigmundPNAS2001}. Specifically, in public goods games with peer punishment, punishers help sustain cooperation by imposing fines on defectors, thereby reducing the payoff advantage of defectors. However, due to the extra cost of punishment, pure cooperators gain an evolutionary advantage over punishers who earn less. This transforms the dilemma to another level, called the second-order free-rider problem, wherein cooperators free-ride on the punishment efforts of others~\cite{HeckathornRS1989,HilbePNAS2014}. Thus, understanding how punishment emerges and evolves under such conditions has attracted significant scientific interest~\cite{HauertScience2007}. To solve this problem, Wang {\it {\rm et al.}}~\cite{WangChaos2025} considered that the multiplication factor depends linearly on both the strategic composition of the entire population and the game groups. They accordingly established two forms of state feedback mechanisms -- global and local -- in the public goods game with punishment. They found that these state feedback mechanisms, especially the local one, can significantly enhance the evolution of cooperation and alleviate the second-order free-rider problem.

However, existing studies mainly focused on the linear strategy frequency dependence of the multiplication factor and explored the conditions required to promote the evolution of cooperation and mitigate the second-order free-riding problem. Indeed, in real-world scenarios, the dependence of multiplication factors on strategy frequency may not be limited to linear forms. Instead, it could also take the quadratic, exponential, or other nonlinear forms~\cite{szolnokipre10,PachecoPloS2015,SantosPNAS2011}.  Therefore, an interesting yet unresolved question arises: if the multiplication factor follows a general form of a frequency-dependent function, what conditions ensure the promotion of cooperation and resolve the second-order free-riding problem?

This work addresses this question by introducing a general frequency-dependent multiplication factor into the public goods game within finite, well-mixed populations. We derive analytical conditions for the evolution of cooperation under weak selection. We find that the necessary condition depends solely on the investment return rate in the full cooperation state of the game group. Once the return value in this state exceeds a critical threshold, natural selection favors cooperation. We also introduce the general frequency-dependent framework to public goods games with peer punishment. Under weak selection, we find that the requested condition for stable cooperation and punishment
depends only on the return rates in the states of full cooperation and full punishment in the interaction group. These conditions are respectively identified. We show that a high return rate in the full cooperation state favors the evolution of cooperation, while a high return value in the full punishment state enhances the evolutionary advantage of punishment. To complete our work, we perform individual-based simulations to verify our theoretical findings.

\section{General frequency-dependent~returns in public goods games}\label{Xsec2-2}\label{Xsec2-2} \label{sec2}

\subsection{Model}\label{Xsec3-2.1}\label{Xsec3-2.1}\label{sec21}

We consider a finite well-mixed population with $N$ individuals. At each time step, $n$ individuals $(2\le n< N)$ are randomly selected from the population to participate in a one-shot public goods game. Each individual in the group can choose a strategy from the set $\mathcal{S}=\left\lbrace C,D \right\rbrace $, where $C$ denotes cooperation and $D$ denotes defection. Cooperators contribute $c$ to the public pool, whereas defectors do not. The total contribution is multiplied by a multiplication factor $r$ and then distributed equally among all $n$ participants. Choosing them randomly, we assume that there are $j_C$ cooperators and $j_D$ defectors among them, therefore $n=j_C+j_D$. Unlike the traditional models~\cite{LiuMMMAS2019,SantosNature2008,SemmannNature2003,SzolnokiJTB2013}, we introduce a general frequency-dependent multiplication factor depending on the strategy composition of the group. Specifically, we define $r$ as a general function dependent solely on the number of cooperators:
\begin{equation}\label{r1}
	r=r(j_{C}).
\end{equation}
Here, the multiplication factor $r$ can characterize the inverse of dilemma strength in the public goods game. 	A larger value of $r$ weakens the dilemma, making cooperation easier to evolve, whereas a smaller $r$ value strengthens the dilemma, favoring defection~\cite{Tanimoto2021}.

Accordingly, the payoffs for a cooperator and a defector within a game group can be written, respectively, as
\begin{equation}\label{pay}
	\Pi_C(j_C,j_D)=\frac{r(j_C)j_C}{n}c-c\quad {\rm and}\quad
	\Pi_D(j_C,j_D)=\frac{r(j_C)j_C}{n}c.
\end{equation}

We consider that the number of individuals adopting a given strategy will evolve in time according to a mutation-selection process combined with the pairwise comparison rule~\cite{SzaboAPS1998,VasconcelosNCC2013}. Specifically, a randomly selected individual $X$ in the population updates its strategy. With probability $\mu$, individual $X$ randomly switches to another strategy from the strategy set of $\mathcal{S}$. With probability $1-\mu,X$ imitates the strategy of a randomly selected individual $Y$ with probability
\[
\frac{1}{1+\exp(-\delta(P_{Y}-P_{X}))}.
\]
Here, $P_X$ and $P_Y$ are the payoff values of individual $X$ and $Y$, respectively. $\delta\ge0$ denotes the selection strength, measuring the role of payoff difference on
the strategy imitation process~\cite{SzaboAPS1998}. In the limit of strong selection $\delta\to\infty$, a player invariably adopts the strategy of the opponent with a higher payoff. Conversely, under the weak selection limit $\delta\to0$, strategy updating becomes essentially random, but strategies yielding higher payoffs are more likely to be imitated.

In the following section, we focus on calculating the average fraction of cooperators in the public goods game with the general frequency-dependent investment returns, and then analyze the conditions for the evolution of cooperation.

\subsection{Theoretical analysis}\label{Xsec4-2.2}\label{sec22}

In this work, we focus on the case of sufficiently small mutation, i.e., $\mu\to 0$. We emphasize that this assumption is mathematically tractable and biologically plausible, as it corresponds to the realistic scenario in which genetic mutations occur relatively rarely compared to selection and drift~\cite{HauertScience2007,SigmundNature2010}. Under this limit, the population spends almost all of its time in one of two homogeneous states in which
all individuals adopt the cooperation or defection strategy~\cite{FudenbergJTB2004}. Thus, the long-term evolutionary dynamics of the population can be described by an embedded Markov chain that captures transitions between these two homogeneous states. Accordingly, the transition probabilities between the two homogeneous states are characterized by the following transition matrix
\begin{equation}\label{mat1}
	\begin{matrix}
		& ~~~ALL C & ALL D  \\
		\begin{matrix}
			ALL C \\
			ALL D \\
		\end{matrix}
		&\left(\  \begin{matrix}
			1-\mu\rho_{DC} \\
			\mu\rho_{CD}\\
		\end{matrix}
		\right. & \left.
		\begin{matrix}
			\mu\rho_{DC} \\
			1-\mu\rho_{CD}\\
		\end{matrix}\ \right),
	\end{matrix}
\end{equation}
where $ALLC$ ($ALLD$) denotes the homogeneous state in which each individual in the population adopts cooperation (defection). $\rho_{uv}$  ($u\ne v$ and $u,v\in\mathcal{S}$) is the fixation probability that the sole mutant who adopts strategy $u$ takes over the entire population consisting of $v$-individuals~\cite{SigmundNature2010}, given by
\begin{equation}\label{fixed}
	\rho_{uv}=\frac{1}{1+\sum_{j=1}^{N-1}\exp(\delta\sum_{N_u=1}^{j}(P_{vu}(N_u)-P_{uv}(N_u)))}.
\end{equation}
Here, $P_{vu}(N_u)$ and $P_{uv}(N_u)$ are the average payoff values of individuals with strategy $v$ and with strategy $u$, respectively, within a population composed of $N_u$ $u$-strategists and $N_v=N-N_u$ $v$-strategists. They are given by
\begin{equation}\label{average1}
	P_{uv}(N_u)=\sum_{n_u=0}^{n-1}\frac{\binom{N_u-1}{n_u}\binom{N-N_u}{n-n_u-1}}{\binom{N-1}{n-1}} \overline{\Pi}_u(n_u+1,n-n_u-1),
\end{equation}
and
\begin{equation}\label{average2}
	P_{vu}(N_u)=\sum_{n_u=0}^{n-1}\frac{\binom{N_u}{n_u}\binom{N-N_u-1}{n-n_u-1}}{\binom{N-1}{n-1}} \overline{\Pi}_v(n_u,n-n_u),
\end{equation}
respectively, where $n_u$ represents the number of $u$-players in the group (excluding the focal individual).
The function $\overline\Pi_i(l,m)$ is the payoff of an $i$-strategist when the group contains $l$ $u$-strategists and $m$ $v$-strategists, given by
\[
\overline{\Pi}_i(l,m)=
\begin{cases}
	\Pi_i(l,m),& \text{ if}\ (u,v)=(C,D), \\
	\Pi_i(m,l),& \text{ if}\ (u,v)=(D,C).
\end{cases}
\]

We note that the stationary distribution of the Markov chain is given by the normalized left eigenvector of the transition matrix~\eqref{mat1} associated with the eigenvalue of 1~\cite{FudenbergJTB2004}, given as
\begin{equation}\label{staionary1}
	\left(f_C,f_D  \right) = \left( \frac{\rho_{CD}}{\rho_{CD}+\rho_{DC}},\frac{\rho_{DC}}{\rho_{CD}+\rho_{DC}}\right).
\end{equation}
Here, $f_i$ ($i\in\mathcal{S}$) measures the average abundance of strategy $i$ during the long evolutionary process~\cite{TraulsrnPRE2006}.

In the following, we derive $f_i$ under the weak selection limit (i.e., $\delta\to0$). This assumption is well-justified, as the payoff from the game typically provides only a minor contribution to strategy updating~\cite{TarnitaPNAS2011}. Since $\sum_{N_u=1}^{N-1}P_{uv}(N_u)$$=\sum_{N_v=1}^{N-1}P_{uv}(N_v)$, the fixation probability can be approximated using the Taylor expansion around $\delta=0$:
\begin{equation}\label{fix}
	\rho_{uv}=
	\frac{1}{N}-\frac{1}{N^2}\left(\sum_{N_v=1}^{N-1}N_v(P_{vu}(N_v)-P_{uv}(N_v)) \right) \delta.
\end{equation}
Therefore, substituting Eq.~{\eqref{fix}} into~\eqref{staionary1}, we get
\begin{equation*}
	\begin{aligned}
		f_C|_{\delta=0}=&f_D|_{\delta=0}=\frac{1}{2},\\
		\frac{df_C}{d\delta}|_{\delta=0}
		=&
		\frac{1}{4}\sum_{N_C=1}^{N-1}(P_{CD}(N_C)-P_{DC}(N_C)),
	\end{aligned}
\end{equation*}
and
\begin{equation*}
	\frac{df_D}{d\delta}|_{\delta=0}=
	-\frac{1}{4}\sum_{N_C=1}^{N-1}(P_{CD}(N_C)-P_{DC}(N_C)),
\end{equation*}
respectively. Consequently, under weak selection, we have
\begin{equation}\label{stationarydistribution1}
	\begin{aligned}
		f_C= \frac{1}{2}+\frac{1}{4}\sum_{N_C=1}^{N-1}(P_{CD}(N_C)-P_{DC}(N_C))\delta,\\
		f_D= \frac{1}{2}-\frac{1}{4}\sum_{N_C=1}^{N-1}(P_{CD}(N_C)-P_{DC}(N_C))\delta.
	\end{aligned}
\end{equation}
Eq.~{\eqref{stationarydistribution1}} shows that the average abundances of strategies are determined by the cumulative payoff difference $\sum_{N_C=1}^{N-1} (P_{CD}(N_C)- P_{DC}(N_C))$. Furthermore, the following proposition provides an explicit form for this key quantity.

\begin{proposition}\label{pro1}
	For the  public goods game, we assume that the multiplication factor is a frequency-dependent function, as defined in Eq.~{\eqref{r1}}. Then
	\begin{equation}
		\sum_{N_C=1}^{N-1}(P_{CD}(N_C)-P_{DC}(N_C))
		=\left( r(n)\frac{N-n}{n}+1-N\right)c,
		\label{Xeqn10}
	\end{equation}
	where $r(n)$ denotes the multiplication factor value in the full-cooperation state of the interaction group.
\end{proposition}

The proof details of this proposition are given in Appendix~{\ref{suppB}}. Intuitively, one would guess that this cumulative payoff difference should depend on the functional expression of the multiplication factor $r(j_C)$. Surprisingly, Proposition~{\ref{pro1}} highlights that the payoff difference depends solely on the boundary investment return $r(n)$ in the full cooperation state, irrelevant to the return values in other strategy configuration states. Furthermore, we find that the average fractions of cooperators and defectors can be determined by this boundary value $r(n)$, given as
\begin{equation}\label{fC}
	\begin{aligned}
		f_C&= \frac{1}{2}+\frac{c}{4}\left( r(n)\frac{N-n}{n}+1-N\right)\delta,\\
		f_D&= \frac{1}{2}-\frac{c}{4}\left( r(n)\frac{N-n}{n}+1-N\right)\delta.
	\end{aligned}
\end{equation}

In the following, according to the calculation of the average frequency of cooperators~\eqref{fC}, we aim to derive the mathematical condition for the evolution of cooperation. To do that, we first present the formal criteria used throughout this paper to determine whether a strategy is evolutionarily successful.
\begin{definition}~\cite{HauertJRSI2024,TarnitaPNAS2011}\label{def1}
	Let $\mathcal{S}$ be the set of strategies,  with $|\mathcal{S}|$ being the number of strategies in this set. Then
	\begin{enumerate}
		\item[(i)] A strategy $u\in\mathcal{S}$ is favored by natural selection if it is more abundant than that in the neutral drift scenario, that is,
		\begin{equation*}
			f_u>\frac{1}{|\mathcal{S}|}.
		\end{equation*}
		\item[(ii)] A strategy $u\in\mathcal{S}$ has an evolutionary advantage over strategy  $v\in\mathcal{S}$ if $u$ is more abundant than $v$ in the long run, that is,
		\begin{equation*}
			f_u>f_v.
		\end{equation*}
	\end{enumerate}
	\label{Xenun9}
\end{definition}

In the following, we present the first main conclusion in this paper, according to Eq.~{\eqref{fC}} and Definition~{\ref{def1}}.

\begin{theorem}\label{them1}
	For the public goods game, we assume that the multiplication factor is frequency-dependent, as defined in Eq.~{\eqref{r1}}.  Then, cooperation is favored by natural selection if
	$r(n)>(N-1)n/(N-n)$.
\end{theorem}

\begin{proof}
	The proof directly follows from Eq.~{\eqref{fC}}.
\end{proof}

\begin{remark}
	This theorem demonstrates that cooperation is favored by selection whenever the investment return rate in the full-contribution state within the game group is sufficiently large. Besides, according to Eq.~{\eqref{fC}}, we find that the abundance of cooperation increases monotonically with this boundary return value.
	
	In particular, when the multiplication factor is constant, i.e., $r(j_C) \equiv r$, the above derived critical condition simplifies to  $r > n(N-1)/(N-n)$~\cite{HauertJRSI2024}. In the limit of an infinite population ($N \to\infty$), this condition further reduces to the well-known classical rule $r >n$~\cite{HauertPRB2006}.
	\label{Xenun3}
\end{remark}

\subsection{Numerical simulations}\label{Xsec5}

We now provide numerical examples to verify the above theoretical results.
\begin{figure*}
	\centerline{\includegraphics[width=0.7\textwidth]{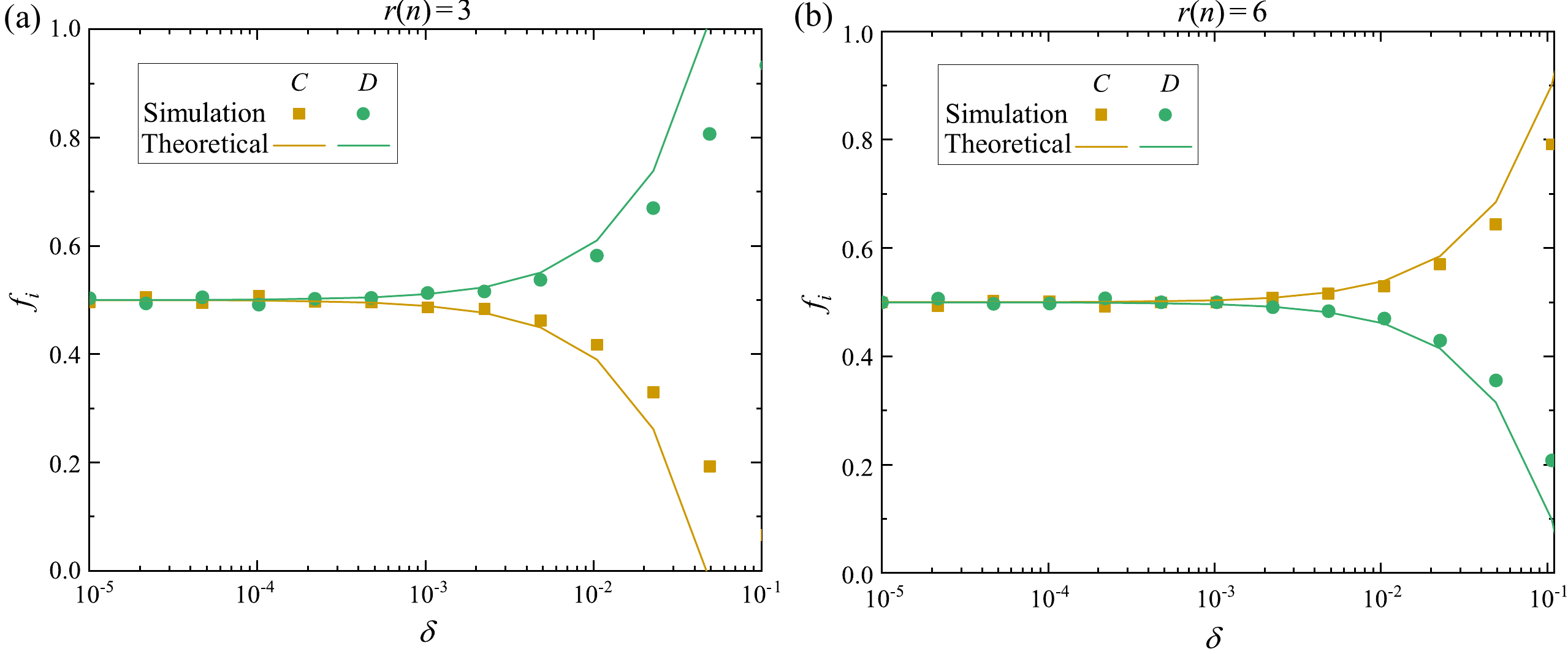}}
	\caption{Average abundance of strategies as a function of the intensity of selection in public goods games with the frequency-dependent return described by Eq.~{\eqref{r_1}}.
		Lines represent theoretical predictions from Eq.~{\eqref{fC}}. Symbols denote the results of individual-based simulations, where each data point is achieved by averaging over $15$ independent realizations and each player takes $10^8$ updates on average in each realization. Parameters: $N=100,n=5,\mu=0.001$, and $r'=3$.}
	\label{fig1}
\end{figure*}
\begin{figure*}
	\centerline{\includegraphics[width=0.7\textwidth]{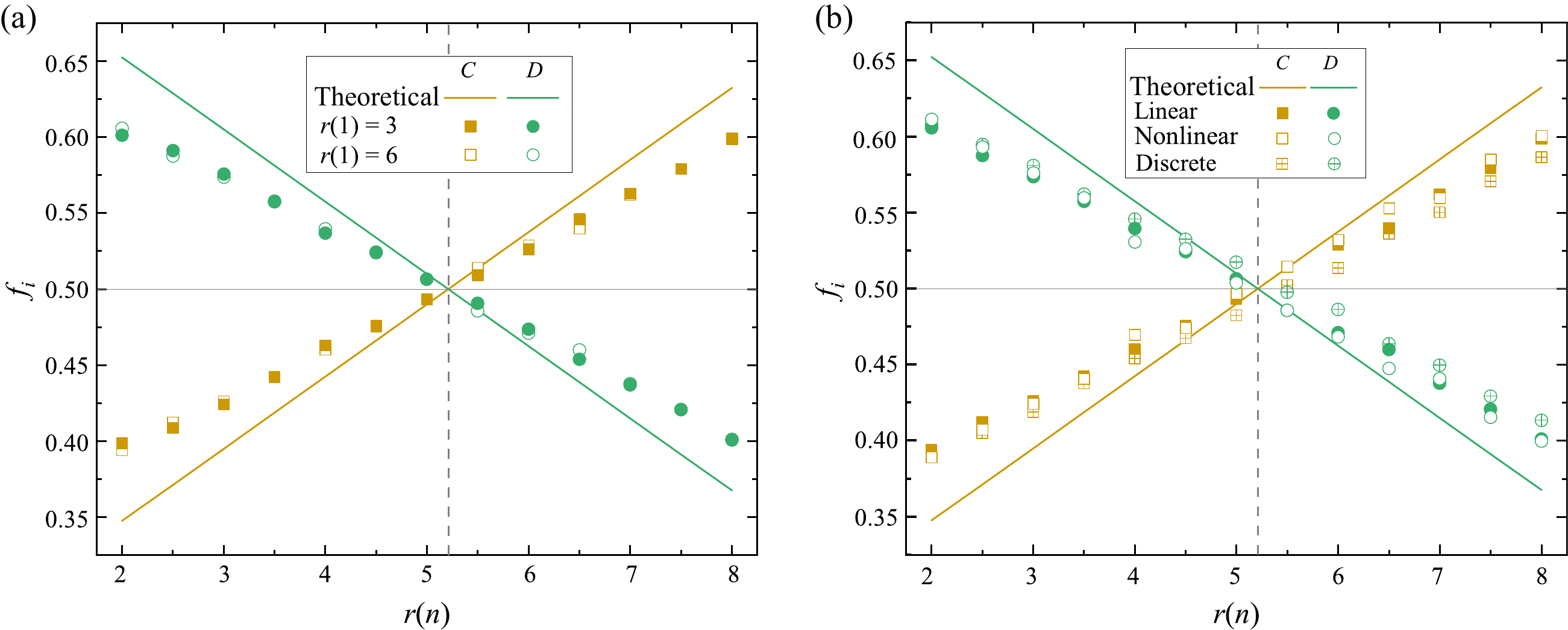}}
	\caption{Average abundance of strategies as a function of the return value $r(n)$ in  the local full-cooperation state. Panel~(a) shows the average abundance of strategies for two different values of $r'$. Panel~(b) presents the average abundance of strategies for different frequency-dependent functions: the linear function~\eqref{r_1}, the nonlinear function~\eqref{r22}, and the piecewise function~\eqref{r0} with $r_1=r_0+k$. Colored solid lines show the corresponding theoretical predictions. Black dashed lines indicate the theoretical threshold value at which cooperation and defection exhibit equal long-run abundance, derived in  Theorem~{\ref{them1}}. Symbols denote the results of individual-based simulations, where each data point is the average over $15$ independent realizations of initial conditions and each player takes $10^8$ updates on average in each realization.
		Parameters: $N=100,n=5,m=5,c=1,\delta=0.01,\mu=0.001, r_0=2$, and $r_2=1.1$.}
	\label{fig2}
\end{figure*}
We first consider the public goods game with a linear frequency-dependent multiplication factor, given by
\begin{equation}\label{r_1}
	r(j_C) = r' + \frac{j_C-1}{n-1}(r''-r'),
\end{equation}
where $r'$ ($r''$) is the investment return when a single (all) individual contributes to the common pool.
If $r'<r''$, this represents economies of scale, whereas $r'>r''$ indicates diminishing returns~\cite{HauertJRSI2024}.
We then have $r(n)=r''$. Fig.~{\ref{fig1}} depicts the average abundance of strategies as a function of the selection intensity in the public goods game with a linear frequency-dependent return. Under a small $r(n)$  (shown in Fig.~{\ref{fig1}}(a)), defectors become more advantageous, while the frequency of cooperators declines with the selection intensity. In contrast, under a large $r(n)$, the frequency of cooperators increases with the selection intensity and cooperators can prevail over defectors.
Importantly, Fig.~{\ref{fig2}}(a) provides a more intuitive verification that higher values of $r(n)$ promote cooperation, and further demonstrates that cooperation is favored when $r(n)>(N-1)n/(N-n)$~\cite{HauertJRSI2024}.

To verify the decisive role of the investment return under the full-cooperation state of the game group, $r(n)$, in the evolution of cooperation, we further examine two other forms of the function for the frequency-dependent multiplication factor: a nonlinear function and a piecewise function. Specifically, the nonlinear frequency-dependent multiplication factor is given by
\begin{equation}\label{r22}
	r(j_C) = r_{0} + \frac{k}{\exp\left( \frac{j_{C}}{n}-1\right) }.
\end{equation}
The piecewise multiplication factor is defined as
\begin{equation}\label{r0}
	r(j_C) =
	\begin{cases}
		r_1,  & {\rm if}\ {j_C\ge m}, \\
		r_2, &  {\rm otherwise},
	\end{cases}
\end{equation}
where $r_1,r_2>1,m\le n$ and $m\in\mathbb{Z}$.
When the conditions $1<r_2<(N-1)n/(N-n)<r_1$ hold and $m$ is sufficiently large, this function represents an extreme case in which a substantial return is achieved only when a majority of individuals in the group contribute, reflecting an incentive for perfect collective action. Fig.~{\ref{fig2}}(b) compares the evolutionary outcomes for these two function forms together with the previously introduced linear form~\eqref{r_1}. The results show that the abundance of cooperation increases with $r(n)$, and this trend holds regardless of the specific functional form of the multiplication factor. Once $r(n)$ exceeds a critical threshold (indicated by the black dashed line), cooperation is favored and has an evolutionary advantage. These findings provide robust numerical support for the theoretical predictions.

\section{General frequency-dependent returns in public goods games with punishment}\label{Xsec6}\label{Xsec4}\label{sec3}
In this section, we introduce the general frequency-dependent multiplication factor into the public goods game with peer punishment, and then analyze the conditions under which cooperation can evolve and the second-order free-rider problem can be mitigated.

\subsection{ Model extension with punishment}\label{Xsec7}\label{Xsec5}\label{sec31}

In the finite population playing the public goods games, we now introduce a third strategy: a punisher ($P$), who contributes to the pool and punishes all defectors in the group. To be specific, each punisher not only contributes $c$ to the public pool, but also imposes a fine of $\beta$ (punishment fine) on each defector on a personal cost $\gamma$ (punishment cost) with $\beta>\gamma>0$.

The strategy set is now enlarged to $\mathcal{S}=\left\lbrace C,D,P\right\rbrace $. Accordingly, we construct the general frequency-dependent multiplication factor in public goods games with peer punishment, which is then described by a bivariate  function
\begin{equation}\label{mu1}
	r=r(j_{C},j_{P}).
\end{equation}
Here, $j_P$ is the number of punishers in the interaction group.
Therefore, the payoffs of cooperators, defectors, and punishers within a game group are written as
\begin{equation}\label{piC}
	\Pi_C(j_C,j_D,j_P)=\frac{r(j_C,j_P)}{n}(j_C+j_P)c-c,
\end{equation}
\begin{equation}\label{piD}
	\Pi_D(j_C,j_D,j_P)=\frac{r(j_C,j_P)}{n}(j_C+j_P)c-\beta j_P,
\end{equation}
and
\begin{equation}\label{piP}
	\Pi_P(j_C,j_D,j_P)=\frac{r(j_C,j_P)}{n}(j_C+j_P)c-c-\gamma j_D,
\end{equation}
respectively. The process for strategy updating remains identical to that described in Section~{\ref{sec21}}. It is a mutation-selection process, governed by the pairwise comparison rule with a rare mutation rate $\mu$.

In the following, we analyze the  conditions under which the introduced frequency-dependent multiplication factor~\eqref{mu1} promotes cooperation and mitigates the second-order free-rider problem.

\subsection{Theoretical analysis}\label{sec32}

First, we calculate the stationary distribution of strategies for the public goods game with punishment, incorporating a general frequency-dependent investment return.

For a rare mutation, the population will never contain more than two different strategies simultaneously, and the population can evolve into a homogeneous state where all individuals adopt the same strategy because the time between two mutation events is long enough~\cite{HauertScience2007}. This allows us to approximate the evolutionary dynamics by means of an embedded Markov chain whose states correspond to the different homogeneous states of the population $\{ ALLC,ALLD,ALLP \} $, characterized by the following transition matrix:

\begin{equation}\label{matrix}	\small
	\setlength{\arraycolsep}{3pt}
	\begin{matrix}
		& ~~~ALL C & ALL D& ALL P  \\
		\begin{matrix}
			ALL C\\
			ALL D\\
			ALL P\\
		\end{matrix}&\left(\begin{matrix}
			1-\frac{\mu\rho_{DC}}{2}-\frac{\mu\rho_{PC}}{2}\\
			\frac{\mu\rho_{CD}}{2}\\
			\frac{\mu\rho_{CP}}{2}
		\end{matrix}
		\right. & \left.
		\begin{matrix}
			\frac{\mu\rho_{DC}}{2}\\
			1-\frac{\mu\rho_{CD}}{2}-\frac{\mu\rho_{PD}}{2}\\
			\frac{\mu\rho_{DP}}{2}\\
		\end{matrix}
		\right. & \left.
		\begin{matrix}
			\frac{\mu\rho_{PC}}{2}\\
			\frac{\mu\rho_{PD}}{2}\\
			1-\frac{\mu\rho_{CP}}{2}-\frac{\mu\rho_{DP}}{2}	 \\
		\end{matrix}\ \right), \\
	\end{matrix}
\end{equation}
where the fixation probability $\rho_{uv}$ ($u,v\in\mathcal{S}$) is given by  Eq.~{\eqref{fixed}}.
In a population of $N_u$ $u$-strategists and $N_v = N - N_u$ $v$-strategists, the average payoff values of a $u$-strategist and a $v$-strategist are

\begin{equation}\label{average3}
	P_{uv}(N_u)=\sum_{n_u=0}^{n-1}\frac{\binom{N_u-1}{n_u}\binom{N-N_u}{n-n_u-1}}{\binom{N-1}{n-1}} \overline{\Pi}_u(n_u+1,n-n_u-1,0),
\end{equation}
and
\begin{equation}\label{average4}
	P_{vu}(N_u)=\sum_{n_u=0}^{n-1}\frac{\binom{N_u}{n_u}\binom{N-N_u-1}{n-n_u-1}}{\binom{N-1}{n-1}} \overline{\Pi}_v(n_u,n-n_u,0),
\end{equation}
where $w$ denotes the strategy in the strategy set distinct from $u$ and $v$. Here, the function $\overline\Pi_i(l,m,k)$ gives the payoff of an $i$-strategist in a group containing $l$ $u$-strategists, $m$ $v$-strategists, and $k$ $w$-strategists. Specifically,
\[
\overline\Pi_i(l,m,k)=
\begin{cases}
	\Pi_i(l,m,k), & \text{if } (u,v,w)=(C,D,P),\\[2pt]
	\Pi_i(m,l,k), & \text{if } (u,v,w)=(D,C,P),\\[2pt]
	\Pi_i(l,k,m), & \text{if } (u,v,w)=(C,P,D),\\[2pt]
	\Pi_i(k,l,m), & \text{if } (u,v,w)=(D,P,C),\\[2pt]
	\Pi_i(m,k,l), & \text{if } (u,v,w)=(P,C,D),\\[2pt]
	\Pi_i(k,m,l), & \text{if } (u,v,w)=(P,D,C).
\end{cases}
\]
The stationary distribution of the matrix~\eqref{matrix} is given by
\begin{equation}\label{stationary2}
	\begin{pmatrix}
		f_{C}\\
		f_{D}\\
		f_{P}
	\end{pmatrix}=
	\frac{1}{M}\begin{pmatrix}
		\rho_{CD}\rho_{CP}+\rho_{CD}\rho_{DP}+\rho_{PD}\rho_{CP}\\
		\rho_{DC}\rho_{CP}+\rho_{DC}\rho_{DP}+\rho_{PC}\rho_{DP}\\
		\rho_{DC}\rho_{PD}+\rho_{PC}\rho_{CD}+\rho_{PC}\rho_{PD}
	\end{pmatrix},
\end{equation}
which captures the average fraction of each strategy over long-term evolution. Here, $M$ is a normalization constant ensuring $\sum_{u\in\mathcal{S}}f_u=1$, given by
\begin{equation*}
	\begin{aligned}
		M=&\,\rho_{CD}\rho_{CP}+\rho_{CD}\rho_{DP}+\rho_{PD}\rho_{CP}
		+\rho_{DC}\rho_{CP}+\rho_{DC}\rho_{DP}\\&+\rho_{PC}\rho_{DP}
		+\rho_{DC}\rho_{PD}+\rho_{PC}\rho_{CD}+\rho_{PC}\rho_{PD}.
	\end{aligned}
\end{equation*}

Correspondingly,
\begin{equation*}
	f_C|_{\delta=0}= f_D|_{\delta=0}= f_P|_{\delta=0}=\frac{1}{3},
\end{equation*}
\begin{equation*}
	\begin{split}
		\frac{df_C}{d\delta}|_{\delta=0}=&
		\frac{1}{9}
		\Big(\sum_{N_C=1}^{N-1}\big( P_{CD}(N_C)-P_{DC}(N_C)
		+ P_{CP}(N_C)\\
		&-P_{PC}(N_C)\big)\vphantom{\sum_{N_C=1}^{N-1}}\Big),\\
		\frac{df_D}{d\delta}|_{\delta=0}=&
		\frac{1}{9}
		\Big(\sum_{N_D=1}^{N-1}\big( P_{DC}(N_D)-P_{CD}(N_D)
		+ P_{DP}(N_D)\\
		&-P_{PD}(N_D)\big)\vphantom{\sum_{N_C=1}^{N-1}}\Big),
	\end{split}
\end{equation*}
and
\begin{equation*}
	\begin{aligned}
		\frac{df_P}{d\delta}|_{\delta=0}=&
		\frac{1}{9}
		\Big(\sum_{N_P=1}^{N-1}\big( P_{PC}(N_P)-P_{CP}(N_P)
		+ P_{PD}(N_P)\\&-P_{DP}(N_P)\big)\vphantom{\sum_{N_C=1}^{N-1}}\Big).
	\end{aligned}
\end{equation*}
Therefore, under the assumption of $\delta\to 0$, the average fractions of cooperators, defectors, and punishers in the population are given as
\begin{equation}\label{stationary}
	\begin{aligned}
		f_C=&\frac{1}{3}+\frac{1}{9}\left( \sum_{N_C=1}^{N-1}\left( P_{CD}(N_C)-P_{DC}(N_C)\right) \right. \\&\left.
		+\sum_{N_C=1}^{N-1}\left( P_{CP}(N_C)-P_{PC}(N_C)\right)\right)\delta, \\
		f_D=&\frac{1}{3}-\frac{1}{9}\left(	\sum_{N_C=1}^{N-1}\left( P_{CD}(N_C)-P_{DC}(N_C)\right)  \right. \\&\left.
		+\sum_{N_P=1}^{N-1}\left( P_{PD}(N_P)-P_{DP}(N_P)\right)\right)\delta,\\
		f_P=&\frac{1}{3}+\frac{1}{9}\left(\sum_{N_P=1}^{N-1}\left( P_{PD}(N_P)-P_{DP}(N_P)\right)\right. \\&\left.
		-\sum_{N_C=1}^{N-1}\left( P_{CP}(N_C)-P_{PC}(N_C)\right)\right)\delta,
	\end{aligned}
\end{equation}
respectively. Obviously, Eq.~{\eqref{stationary}} are determined by three payoff difference sums which are: $\sum_{N_C=1}^{N-1}(P_{CD}(N_C)-P_{DC}(N_C))$, $\sum_{N_C=1}^{N-1}\left( P_{CP}(N_C)-P_{PC}(N_C)\right)$, and $\sum_{N_P=1}^{N-1} ( P_{PD}(N_P)-P_{DP}(N_P))$~\cite{SuiPLA2017}{.} We calculate these three terms as shown in the following proposition.

\begin{proposition}\label{pro2}
	For the public goods game with punishment, we assume that the multiplication factor is a frequency-dependent function, as defined in Eq.~{\eqref{mu1}}. Then
	\begin{equation}\label{PCD1}
		\begin{aligned}
			\sum_{N_C=1}^{N-1}(P_{CD}(N_C)-P_{DC}(N_C))
			=\left( r(n,0)\frac{N-n}{n}+1-N\right)c,
		\end{aligned}
	\end{equation}
	\begin{equation}\label{PCD2}
		\begin{aligned}
			&\sum_{N_C=1}^{N-1}\left( P_{CP}(N_C)-P_{PC}(N_C)\right)=\frac{N-n}{n}c(r(n,0)-r(0,n)),
		\end{aligned}
	\end{equation}
	and
	\begin{equation}\label{PCD3}
		\begin{aligned}
			&\sum_{N_P=1}^{N-1}\left( P_{PD}(N_P)-P_{DP}(N_P)\right)\\
			&\quad=c\left( r(0,n)\frac{N-n}{n}+1-N\right)+(\beta-\gamma)\frac{N(n-1)}{2},
		\end{aligned}
	\end{equation}
	where $r(n,0)$ and $r(0,n)$ are the multiplication factor values in the full-cooperation and full-punishment states, respectively.
\end{proposition}

The proof of Proposition~{\ref{pro2}} is provided in Appendix~{\ref{suppC}}. This proposition indicates that the three cumulative payoff differences within Eq.~{\eqref{stationary}} are determined solely by the investment return values in the group state of full-cooperation $r(n,0)$ and full-punishment $r(0,n)$. Accordingly,
\begin{equation}\label{sta2}
	\begin{split}
		f_C=&\,\frac{1}{3}+\frac{c}{9}\left((2r(n,0)-r(0,n))\frac{N-n}{n}+ 1-N\right)\delta,\\
		f_D=&\,\frac{1}{3}+\frac{1}{9}\Big(-\left( r(n,0)+r(0,n)\right) \frac{N-n}{n}c 
		-2\left(1-N\right)c\\
		&-(\beta-\gamma)\frac{N(n-1)}{2}\Big)\delta,\\
		f_P=&\,\frac{1}{3}+\frac{1}{9}\Big(\left(  \left( 2r(0,n)-r(n,0)\right)\frac{N-n}{n}+1-N\right) c 
		\\&+(\beta-\gamma)\frac{N(n-1)}{2}\Big) \delta.
	\end{split}
\end{equation}
Based on Eq.~{\eqref{sta2}}, we have the following conclusion.

\begin{theorem}\label{them2}
	For the public goods game with punishment, we assume the multiplication factor is frequency-dependent, as defined in Eq.~{\eqref{mu1}}. Then the following results hold:
	\begin{enumerate}
		\item[(i)]
		Cooperation is favored if $r(n,0)>\frac{r(0,n)}{2}+\frac{n(N-1)}{2(N-n)}$. It has an evolutionary advantage over defection when $r(n,0)>\frac{(N-1)n}{N-n} -(\beta-\gamma)\frac{Nn(n-1)}{6c(N-n)}$, and over punishment when
		$ r(n,0)>r(0,n) +(\beta-\gamma)\frac{Nn(n-1)}{6c(N-n)}$.
		\item[(ii)]
		Punishment is favored if $r(0,n)>\frac{r(n,0)}{2}+\frac{n(N-1)}{2(N-n)}-(\beta-\gamma)\frac{Nn(n-1)}{4c(N-n)}$.
		It has an evolutionary advantage over defection when $r(0,n)>\frac{(N-1)n}{N-n} -(\beta-\gamma)\frac{Nn(n-1)}{3c(N-n)}$,
		and over cooperation when $ r(0,n)>r(n,0) -(\beta-\gamma)\frac{Nn(n-1)}{6c(N-n)}$.
	\end{enumerate}
\end{theorem}
\begin{proof}
	The proof directly follows from Eq.~{\eqref{sta2}}.
\end{proof}

\begin{remark}	 Theorem~{\ref{them2}} indicates the evolutionary landscape shaped by the synergy between the feedback mechanism and punishment effect. In particular, the result of Theorem~{\ref{them1}} can be recovered from Theorem~{\ref{them2}} by setting $j_P\equiv 0$.
	\label{Xenun4}
\end{remark}

\begin{figure*}
	\centerline{\includegraphics[width=0.87\textwidth]{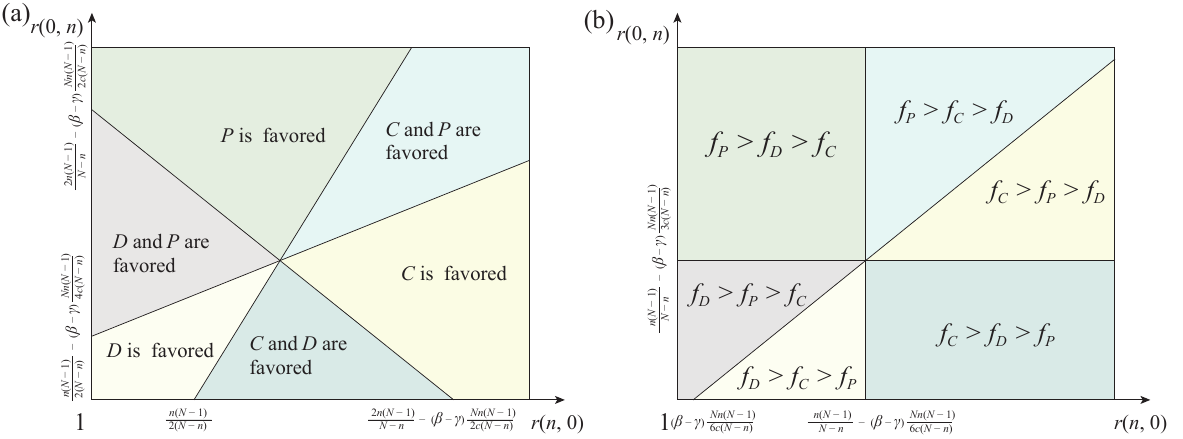}}
	\caption{Schematic of the evolutionary outcomes of strategies in the public goods game with punishment under the general frequency-dependent return.	Panel~(a) shows	the evolutionary diagram under which strategies are favored by natural selection, while panel~(b) presents the evolutionary diagram under which strategies have an evolutionary advantage over others.}
	\label{fig3}
\end{figure*}

According to Theorem~{\ref{them2}}, we know that the evolutionary outcome primarily depends on the investment return values $r(n, 0)$ in the full-cooperation state and $r(0,n)$ in the full-punishment state.
In order to help the reader to understand the evolutionary outcomes for different parameter regions more easily, we present an illustrative plot of the dynamical regimes in the parameter space ($r(n,0),r(0,n)$), as shown in  Fig.~{\ref{fig3}}. Specifically, when either $r(n,0)$ or $r(0,n)$ is small, the system tends toward a defection dominated state.
A large $r(n,0)$ promotes cooperation, whereas a large $r(0,n)$ favors the emergence of punishment. Furthermore, if both values $r(n,0)$ and $r(0,n)$ are high, cooperation and punishment can simultaneously be favored, as shown by the light blue region in Fig.~{\ref{fig3}}(a).

A key finding that cannot be observed in the  traditional model with fixed investment return  is that a high $r(0,n)$ can directly
counteract the negative effect induced by a high punishment cost on the evolution of punishment. This is clearly reflected in the condition for punishment to be favored: $r(0,n)> \frac{r(n,0)}{2} +\frac{n(N-1)}{2(N-n)} - (\beta - \gamma) \frac{Nn(n-1)}{4c(N-n)}$.
The condition includes a negative term related to the effective cost of punishment $(\beta - \gamma)$. Thus, even if the intrinsic punishment cost $\gamma$ is high, a sufficiently large return value $r(0,n)$ can still make the inequality hold and promote the evolution of punishment.
This indicates that in the public goods game with punishment incorporating a frequency-dependent multiplication factor, punishment can evolve through two distinct pathways: one relies on a high fine $\beta$ to deter defectors, and the other leverages a high investment return in the full-punishment state of the group. This finding provides a deeper understanding of how costly punishment can emerge and be sustained.

We next consider a special case where the general frequency-dependent
multiplication factor depends solely on the number of contributors including cooperators and punishers in the group, and we have $r=r(j_C,j_P)=r(j_{C}+j_{P})$. Under this assumption, we thus have  $r_n\triangleq r(n,0)=r(0,n)$ and obtain the following corollary.

\begin{corollary}\label{cor1}
	Assume that the general frequency-dependent
	multiplication factor depends solely on the number of contributors in the group. Then the following results hold:
	\begin{enumerate}
		\item[(i)]
		Cooperation is favored if $r_n>\frac{(N-1)n}{N-n}$. It has an evolutionary advantage over defection when $r_n>\frac{(N-1)n}{N-n}-(\beta-\gamma)\frac{Nn(n-1)}{6c(N-n)}$.
		\item[(ii)]
		Punishment is favored if $r_n>\frac{(N-1)n}{N-n}-(\beta-\gamma)\frac{Nn(n-1)}{2c(N-n)}$. It has an evolutionary advantage over defection when $r_n>\frac{(N-1)n}{N-n}-(\beta-\gamma)\frac{Nn(n-1)}{3c(N-n)}$, and over cooperation for any $r_n$.
	\end{enumerate}
\end{corollary}

\begin{remark}
	Corollary~{\ref{cor1}} shows that when returns depend equally on the frequencies of cooperators and punishers, punishment exhibits an inherent advantage over cooperation.
	
	Surprisingly, the critical condition for cooperation to be favored, $r_n > n(N-1)/(N-n)$, is identical to the condition in the public goods game without punishment (see Theorem~{\ref{them1}}). This means that, in a system with this special feedback, the introduction of punishment does not alter the condition for cooperation being favored by selection.
	\label{Xenun5}
\end{remark}

\subsection{Individual-based simulations}\label{sec33}
To verify our theoretical results, we first consider the following linear frequency-dependent function, given by

\begin{figure*}
	\centerline{\includegraphics[width=0.8\textwidth]{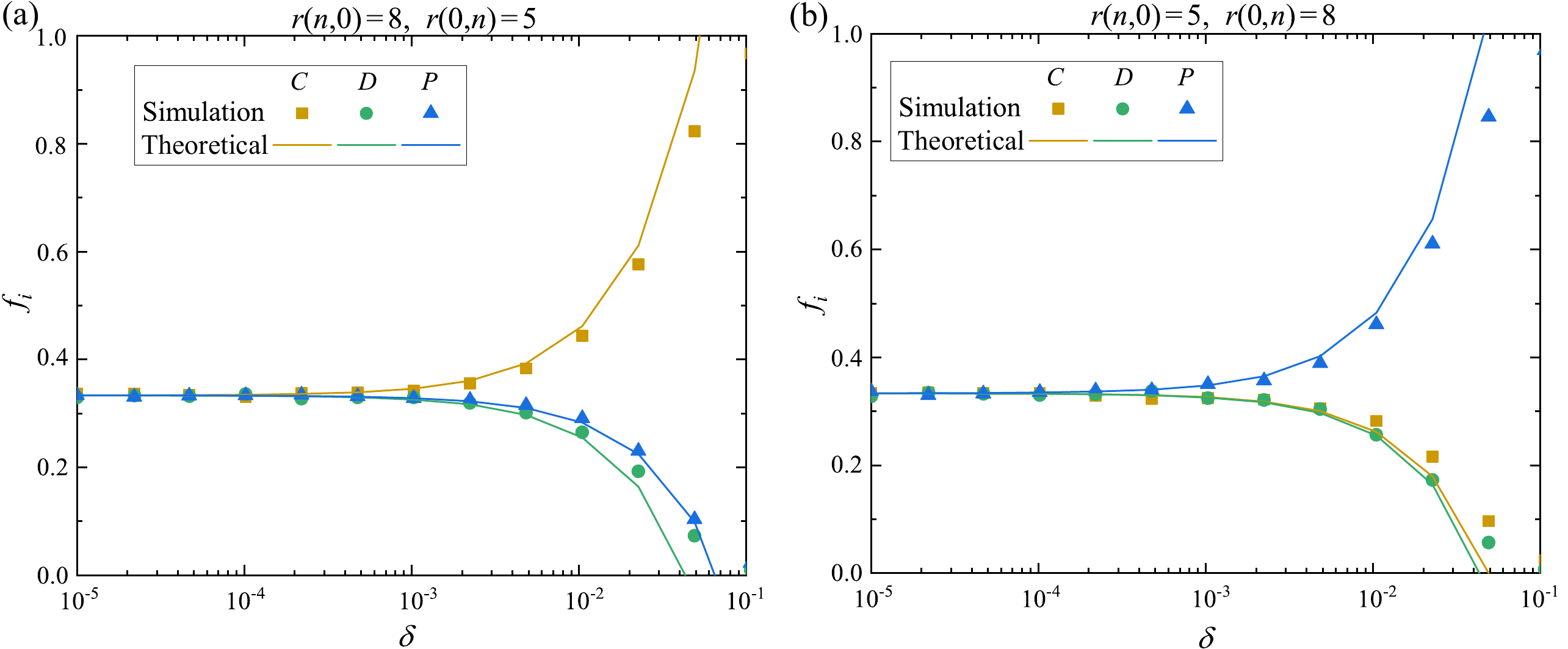}}
	\caption{Average abundance of strategies as a function of the intensity of selection in the public goods games with punishment under the frequency-dependent return function~\eqref{r11}.
		Lines represent theoretical predictions derived from Eq.~{\eqref{sta2}}. Symbols denote individual-based simulation  results, where each data point is achieved by averaging over $15$ independent realizations of initial conditions and each player takes $10^8$ updates on average in each realization.
		Parameters: $N=100,n=5,\beta=0.1,\gamma=0.01,\mu=0.001$, and $r_0=2$.}
	\label{fig4}
\end{figure*}
\begin{figure*}
	\centerline{\includegraphics[width=0.8\textwidth]{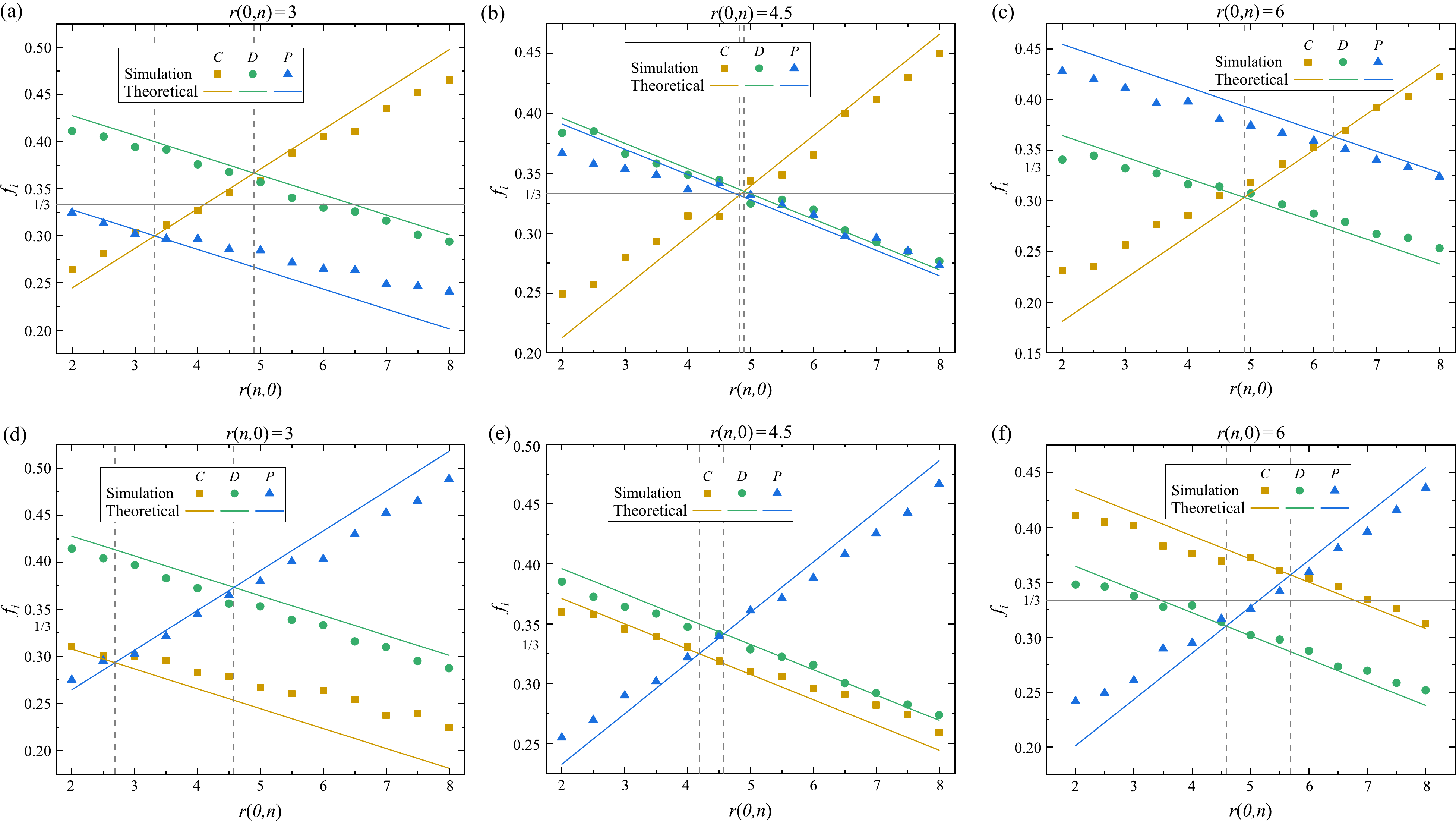}}
\caption{Average abundance of strategies in the public goods game with punishment as a function of the return values $r(n,0)$ (first row) and $r(0,n)$ (second row).Symbols represent the simulated results for the average frequencies of strategies. Colored solid lines show the corresponding theoretical predictions. Black dashed lines indicate the theoretical threshold value at which each strategy exhibits identical long-run abundance, derived in Theorem~{\ref{them2}}. The simulated results are averaged over $15$ independent realizations and each player takes  $10^8$ updates on average  in each realization. Parameters: $N=100,n=5,\delta=0.01,\beta=0.1,\gamma=0.01,\mu=0.001$, and $ r_0=2$.}
	\label{fig5}\vspace{4pt}
\end{figure*}

\begin{figure}
	\centerline{\includegraphics[width=0.4\textwidth]{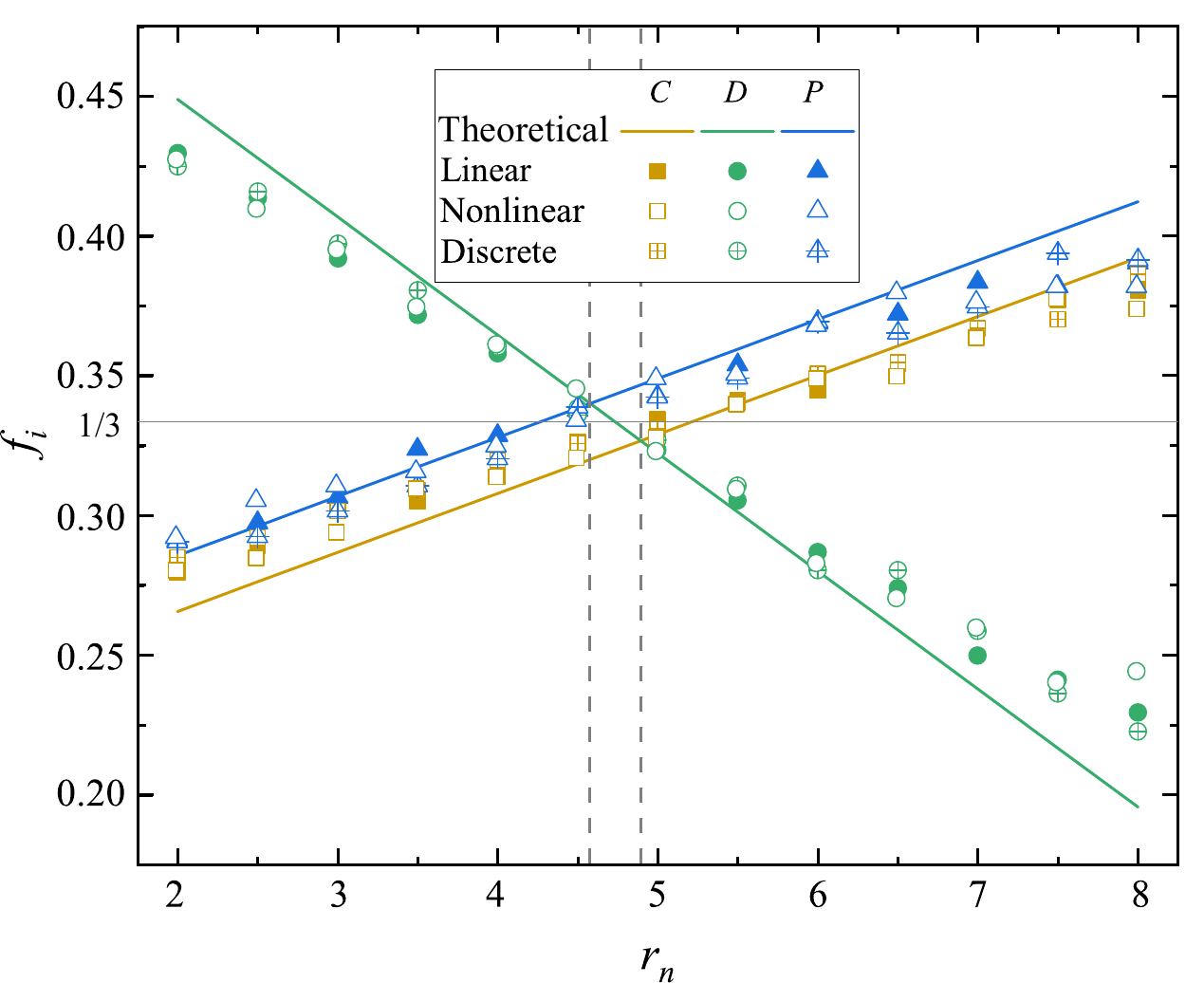}}
	\caption{Average abundance of strategies as a function of the return value $r_n$ in the local full-contribution state, shown for three frequency-dependent functions: the linear function~\eqref{r2}, the nonlinear function~\eqref{r222}, and the piecewise function~\eqref{r01} with $r_1=r_0+k$.Symbols represent the simulated results for the average frequencies of the three strategies.	Colored solid lines show the corresponding theoretical predictions. Black dashed lines indicate the theoretical threshold value at which each strategy exhibits identical long-run abundance, derived in Corollary~{\ref{cor1}}. The simulated results are averaged over $15$ independent realizations of initial conditions and each player takes $10^8$ updates on average in each realization. Parameters: $N=100,n=5,m=5,\delta=0.01,\beta=0.1,\gamma=0.01,\mu=0.001, r_0=2$, and $r_2=1.1$.}
	\label{fig6}
\end{figure}

\begin{equation}\label{r11}
	r(j_C,j_P) = r_0 + k_1\frac{j_C}{n}+k_2\frac{j_P}{n},
\end{equation}
where $r_{0}>1$ is the baseline multiplication factor. $k_1$ and $k_2$ represent the feedback strengths of cooperators and punishers, respectively. Here, we have $r(n,0)=r_0+k_1$ and $r(0,n)=r_0+k_2$.
 Fig.~{\ref{fig4}} shows the average abundances of strategies in dependence on the intensity of selection under this frequency-dependent multiplication factor. We can observe that the simulated results align closely with the theoretical predictions, even when the selection is not extremely weak. Moreover, when $r(n,0)$ is larger than $r(0,n)$, the average frequency of cooperators exceeds that of punishers and defectors (Fig.~{\ref{fig4}}(a)). Otherwise, punishers become dominant while cooperators and defectors are less abundant ( Fig.~{\ref{fig4}}(b)). This confirms that the relative magnitudes of $r(n,0)$ and $r(0,n)$ crucially determine the long-term evolution of strategies.

To further clarify the decisive role of boundary returns $r(n,0)$ and $r(0,n)$ in shaping the evolutionary outcomes, we show how the average frequencies of strategies  vary with increasing $r(n,0)$ or $r(0,n)$ in Fig.~{\ref{fig5}}. We can see that the simulated results align with the theoretical prediction obtained by Eq.~{\eqref{sta2}}. Besides, when the frequency-dependent feedback strength is weak, leading to a smaller $r(n,0)$ or $r(0,n)$, both cooperation and punishment are at a disadvantage compared to\break defection, as seen in Fig.~{\ref{fig5}}(a) and (d). As the frequency-dependent feedback strength becomes strong, the average frequencies of both cooperation and punishment increase significantly. Furthermore, when the return value in either the full-cooperation state or the full-punishment state is sufficiently high, cooperation or punishment can become the dominant strategy (top and bottom rows of Fig.~{\ref{fig5}}). This demonstrates that considering frequency-dependent return can mitigate or even reverse the negative effects induced by cooperation and punishment costs, thereby promoting the evolution of both cooperative and punishing\break strategies.

To verify the robustness of these results, we further consider the average frequencies of strategies under various forms of return function. For simplicity, we focus on three function forms: a linear function given by
\begin{equation}\label{r2}
	r(j_C,j_P) = r_{0} + k\frac{j_{C} + j_{P}}{n};
\end{equation}
a nonlinear function defined as
\begin{equation}\label{r222}
	r(j_C,j_P) = r_{0} + \frac{k}{\exp(\frac{j_{C} + j_{P}}{n}-1)};
\end{equation}
and a piecewise function given by
\begin{equation}\label{r01}
	r(j_C,j_P) =
	\begin{cases}
		r_1,  & {\rm if}\ {j_C+j_P\ge m}, \\
		r_2, &  {\rm otherwise},
	\end{cases}
\end{equation}
with $r_1,r_2>1,m\le n$, and $m\in\mathbb{Z}$. Note that under these three functions, the feedback-based return only depends on the number of contributors including cooperators and punishers in the group and we have $r_n\triangleq r(n,0)=r(0,n)$. As shown in Fig.~{\ref{fig6}}, we present the average abundance of each strategy as a function of $r_n$ in the full-contribution state for these three return functions. Despite their distinct mathematical expressions, all functions yield highly consistent strategy abundances because they share the same value of $r_n$. A higher $r_n$ simultaneously promotes the evolution of both cooperative and punishing behaviors. Once $r_n$  exceeds a critical threshold, the average abundances of cooperators and punishers become significantly higher than that of defectors. These results confirm that our conclusions remain robust across different forms of return function.

\section{Conclusion and discussion}\label{Xsec8}\label{Xsec6}

In this work, we have introduced the general frequency-dependent investment returns into the public goods game, where the multiplication factor dynamically changes with the strategic composition of the game group. By means of theoretical analysis, we have derived mathematical conditions for the evolution of cooperation in the limit of weak selection. We have shown that whether cooperation can be favored is primarily determined by the investment return value in the full-cooperation state of the game group. Specifically, a higher return value in such a state consistently promotes the average abundance of cooperators in the long run. Furthermore, we have introduced the general frequency-dependent return into public goods games with peer punishment. We have found that the average abundance of strategies is shaped by the returns in both the full-cooperation and full-punishment states of the group.  Specifically, a higher return rate in the full-cooperation state enhances the evolutionary advantage of cooperation, while a higher return value in the full-punishment state favors the dominance of punishment. Importantly, all analytical predictions have been supported by individual-based simulations. This study elucidates the conditions to promote cooperation and alleviate the second-order free-riding problem under general frequency-dependent returns.

In this paper we focus on how frequency-dependent returns alleviate the second-order free-riding problem and promote cooperation when the peer punishment is considered, but it is also important to consider the explicit goal of optimizing social welfare. As shown in Refs{.} \cite{HanJRSI2024,HanPLR2025}, optimized social welfare does not always align with maximal cooperation due to the hidden costs introduced by punishment.
Therefore, within the framework of general frequency-dependent returns, social welfare could serve as a key metric for future work.
Moreover, in this work we only consider peer punishment.  However, many other sanction mechanisms~\cite{RandNC2011,SigmundPNAS2001,Han2016}, such as institutional punishment and peer/institutional reward, have received widespread attention in the study of social dilemmas.  Investigating whether, and if so how, boundary returns play a decisive role under these alternative incentive mechanisms is another highly promising research direction.

Last, we note that the present work is conducted under the assumptions of a well-mixed population and weak selection.
However, in real-world scenarios, populations are often structured, with interactions typically limited to local neighbors~\cite{AllenNature2017,TarnitaPNAS2011,TanPD2024,PercNJP2012}. Moreover, selection may operate well beyond the limit of weak selection explored here.
A natural and important extension would therefore be to study the general frequency-dependent return in structured populations, where the return depends on the strategic configuration within an individual's local neighborhood, and to examine whether the result remains robust under stronger selection.  Such extensions would significantly enhance the practical relevance and predictive power of the theoretical framework developed here, and we also intend to pursue them in future\break work.

	\section*{Declaration of competing interest}
	The authors declare that they have no known competing financial
	interests or personal relationships that could have appeared to
	influence the work reported in this paper.

	\section*{Acknowledgment}
	This research was supported by the National Natural Science Foundation of China under Grant No. 62473081 and by the National Research, Development and Innovation Office (NKFIH), Hungary under Grant No. K142948.

\section*{Data availability}
No data was used for the research described in the article.

\appendix

\section{\label{suppA}Derivation of some important identities}\nonumber We here show how the following important identities are derived by using \textit{Hockey-stick identity} and its extension: $\sum_{m=r}^{M}\binom{m}{r}=\binom{M+1}{r+1}$  and $\sum_{m=r}^{M}\binom{m}{r}\binom{M-m}{s}=\binom{M+1}{r+s+1}$
to obtain the theoretical results shown in the main text.
Note that $\sum_{k=s}^{t}f(k)=0$ if $t<s$.

Assume that $k,s,n,N\in \mathbb{Z}$, $1\le n\le s\le N-1$,  and $k\le n-1$. Then, using the key identities $\sum_{m=r}^{M}\binom{m}{r}=\binom{M+1}{r+1}$ and $\sum_{m=r}^{M}\binom{m}{r}\binom{M-m}{s}=\binom{M+1}{r+s+1}$, we have
\begin{enumerate}[(i)]
	\item
	\begin{equation*}
		\begin{split}
			&\sum_{s=1}^{N-1}\frac{\binom{s-1}{n-1}}{\binom{N-1}{n-1}}
			=\frac{1}{\binom{N-1}{n-1}}\sum_{s=n}^{N-1}\binom{s-1}{n-1}\\
			=&\frac{1}{\binom{N-1}{n-1}}\sum_{t=n-1}^{N-2}\binom{t}{n-1}
			=\frac{\binom{N-1}{n}}{\binom{N-1}{n-1}}=\frac{N-n}{n}.
		\end{split}
	\end{equation*}
	
	\item 
	When $k\ne0$,
	\begin{equation*}
		\begin{split}
			&\sum_{s=1}^{N-1}\frac{\binom{s}{k}\binom{N-s-1}{n-k-1}}{\binom{N-1}{n-1}}
			=\frac{1}{\binom{N-1}{n-1}}\sum_{s=k}^{N-1}\binom{s}{k}\binom{N-s-1}{n-k-1}\\
			=&\frac{\binom{N}{n}}{\binom{N-1}{n-1}}=\frac{N}{n}.
		\end{split}
	\end{equation*}
	When $k=0$,
	\begin{equation*}
		\begin{split}
			&\sum_{s=1}^{N-1}\frac{\binom{s}{k}\binom{N-s-1}{n-k-1}}{\binom{N-1}{n-1}}
			=\sum_{s=1}^{N-1}\frac{\binom{N-s-1}{n-1}}{\binom{N-1}{n-1}}\\
			=&\sum_{t=1}^{N-1}\frac{\binom{t-1}{n-1}}{\binom{N-1}{n-1}}
			=\frac{N-n}{n}.
		\end{split}
	\end{equation*}
	Similarly, 
	\begin{equation*}
		\sum_{s=1}^{N-1}\frac{\binom{s-1}{k}\binom{N-s}{n-k-1}}{\binom{N-1}{n-1}}
		=\begin{cases}
			\frac{N}{n}, &  {k\ne n-1}, \\
			\frac{N-n}{n},  & {k=n-1}.
		\end{cases}
	\end{equation*}
	
	\item 
	When $k\ne0$, 
	\begin{equation*}
		\begin{split}
			&\sum_{s=1}^{N-1}\frac{\binom{s}{k}\binom{N-s-1}{n-k-1}}{\binom{N-1}{n-1}}\frac{1}{s}\\
			=&
			\frac{1}{k\binom{N-1}{n-1}}\sum_{s=1}^{N-1}\binom{s-1}{k-1}\binom{N-s-1}{n-k-1}\\
			=&
			\frac{1}{k\binom{N-1}{n-1}}\sum_{t=0}^{N-2}\binom{t}{k-1}\binom{N-t-2}{n-k-1}\\
			=&
			\frac{1}{k\binom{N-1}{n-1}}\binom{N-1}{n-1}=\frac{1}{k}.
		\end{split}
	\end{equation*}
	
	\item
	\begin{equation*}
		\begin{split}
			&\sum_{s=1}^{N-1}\sum_{k=0}^{n-1}\frac{\binom{s}{k}\binom{N-s-1}{n-k-1}}{\binom{N-1}{n-1}}k\\
			=&\frac{1}{\binom{N-1}{n-1}} \sum_{s=1}^{N-1}\sum_{k=1}^{n-1}k \binom{s}{k}\binom{N-s-1}{n-k-1}\\
			=&\frac{1}{\binom{N-1}{n-1}} \sum_{k=1}^{n-1}k\binom{N}{n}
			=\frac{N(n-1)}{2}.
		\end{split}
	\end{equation*}
	Furthermore,
	\begin{equation*}
		\begin{split}
			&\sum_{s=1}^{N-1}\sum_{k=0}^{n-1}\frac{\binom{s-1}{k}\binom{N-s}{n-k-1}}{\binom{N-1}{n-1}}k
			=
			\sum_{k=0}^{n-1}\sum_{s=0}^{N-2}\frac{\binom{s}{k}\binom{N-s-1}{n-k-1}}{\binom{N-1}{n-1}}k\\
			=&
			\sum_{s=1}^{N-1}\sum_{k=0}^{n-1}\frac{\binom{s}{k}\binom{N-s-1}{n-k-1}}{\binom{N-1}{n-1}}k -(n-1)=\frac{(N-2)(n-1)}{2}.
		\end{split}
	\end{equation*}
	
	\item
	\begin{equation*}
		\begin{split}
			&\sum_{s=1}^{N-1}\sum_{k=1}^{n-1}\frac{\binom{s}{k}\binom{N-s-1}{n-k-1}}{\binom{N-1}{n-1}}\frac{k(N-s)}{s(n-k)}\\
			=&\frac{1}{\binom{N-1}{n-1}}\sum_{s=1}^{N-1}\sum_{k=1}^{n-1}\binom{s-1}{k-1}\binom{N-s}{n-k}\\
			=&\frac{1}{\binom{N-1}{n-1}}\sum_{k=0}^{n-2}\sum_{s=0}^{N-2}\binom{s}{k}\binom{N-s-1}{n-k-1}\\
			=&
			\frac{1}{\binom{N-1}{n-1}}\sum_{k=0}^{n-2}\sum_{s=k}^{N}\binom{s}{k}\binom{N-s-1}{n-k-1}\\
			=&\frac{1}{\binom{N-1}{n-1}}\sum_{k=0}^{n-2}\binom{N}{n}
			=\frac{N(n-1)}{n}.
		\end{split}
	\end{equation*}
\end{enumerate}

\section{\label{suppB}Proof of Proposition \ref{pro1}}
We substitute the respective payoffs~\eqref{average1} and~\eqref{average2} into $\sum_{N_C=1}^{N-1}(P_{CD}(N_C)$ $-P_{DC}$ $(N_C))$, and obtain 
\begin{align*}
		&\sum_{N_C=1}^{N-1}(P_{CD}(N_C)-P_{DC}(N_C))\\
		=&\sum_{N_C=1}^{N-1}\sum_{n_C=0}^{n-1}\frac{\binom{N_C-1}{n_C}\binom{N-N_C}{n-n_C-1}}{\binom{N-1}{n-1}} \Pi_C(n_C+1,n-n_C-1)
		\\&
		- 
		\sum_{N_C=1}^{N-1}\sum_{n_C=0}^{n-1} \frac{\binom{N_C}{n_C}\binom{N-N_C-1}{n-n_C-1}}{\binom{N-1}{n-1}}\Pi_D(n_C,n-n_C)\\
		=&
		\sum_{N_C=1}^{N-1}\sum_{n_C=0}^{n-1}\frac{\binom{N_C-1}{n_C}\binom{N-N_C}{n-n_C-1}}{\binom{N-1}{n-1}}\left(\frac{r(n_C+1)(n_C+1)}{n}c\right. \\&\left.
		-c\vphantom{\sum_{N_C=1}^{N-1}} \right) 
		-
		\sum_{N_C=1}^{N-1}\sum_{n_C=0}^{n-1} \frac{\binom{N_C}{n_C}\binom{N-N_C-1}{n-n_C-1}}{\binom{N-1}{n-1}}\left(\frac{r(n_C)n_C}{n}c \right) \\
		=&
		\sum_{N_C=1}^{N-1}\left( \frac{\binom{N_C-1}{n-1}\binom{N-N_C}{0}}{\binom{N-1}{n-1}}\left(\frac{r(n)n}{n}c-c \right)\right. \\&\left.
		+ \sum_{k=1}^{n-1}\frac{\binom{N_C-1}{k-1}\binom{N-N_C}{n-k}}{\binom{N-1}{n-1}}\left(\frac{r(k)k}{n}c-c \right)
		\right)\\& 
		-
		\sum_{N_C=1}^{N-1}\sum_{n_C=0}^{n-1} \frac{\binom{N_C}{n_C}\binom{N-N_C-1}{n-n_C-1}}{\binom{N-1}{n-1}}\left(\frac{r(n_C)n_C}{n}c \right) \\
		=&
		\sum_{N_C=1}^{N-1} \frac{\binom{N_C-1}{n-1}}{\binom{N-1}{n-1}}\left(r(n)c-c \right) 
		+ \sum_{N_C=1}^{N-1} \sum_{k=1}^{n-1}\frac{\binom{N_C}{k}\binom{N-N_C-1}{n-k-1}}{\binom{N-1}{n-1}}\\&
		\cdot\left(\left( \frac{k(N-N_C)}{N_C(n-k)}-1\right)
		\frac{r(k)k}{n}c \right)\\
		&-
		\sum_{N_C=1}^{N-1}\sum_{k=1}^{n-1}\frac{\binom{N_C}{k}\binom{N-N_C-1}{n-k-1}}{\binom{N-1}{n-1}}\frac{k(N-N_C)}{N_C(n-k)}c\\
		=&
		\left(r(n)c-c \right)\sum_{N_C=1}^{N-1} \frac{\binom{N_C-1}{n-1}}{\binom{N-1}{n-1}}-
		c\sum_{N_C=1}^{N-1}\sum_{k=1}^{n-1}\frac{\binom{N_C}{k}\binom{N-N_C-1}{n-k-1}}{\binom{N-1}{n-1}}\\&
		\cdot \frac{k(N-N_C)}{N_C(n-k)}
		+
		\frac{c}{n}\sum_{N_C=1}^{N-1} \sum_{k=1}^{n-1}\frac{\binom{N_C}{k}\binom{N-N_C-1}{n-k-1}}{\binom{N-1}{n-1}}\\
		&\cdot\left(\left( \frac{k(N-N_C)}{N_C(n-k)}-1\right) r(k)k\right).
\end{align*}

Based on the identities derived in Appendix~\ref{suppA}, we have
\begin{equation*}
	\begin{split}
		&\sum_{N_C=1}^{N-1} \sum_{k=1}^{n-1}\frac{\binom{N_C}{k}\binom{N-N_C-1}{n-k-1}}{\binom{N-1}{n-1}}\left(\left( \frac{k(N-N_C)}{N_C(n-k)}-1\right)r(k)k \right)\\
		=&
		\sum_{k=1}^{n-1}  r(k)k \frac{k}{n-k}\left( N \sum_{N_C=1}^{N-1} \frac{\binom{N_C}{k}\binom{N-N_C-1}{n-k-1}}{\binom{N-1}{n-1}}\frac{1}{N_C}
		\right. \\&\left.
		-\sum_{N_C=1}^{N-1} \frac{\binom{N_C}{k}\binom{N-N_C-1}{n-k-1}}{\binom{N-1}{n-1}}\right) 
		- \sum_{k=1}^{n-1}  r(k)k\sum_{N_C=1}^{N-1} \frac{\binom{N_C}{k}\binom{N-N_C-1}{n-k-1}}{\binom{N-1}{n-1}}\\
		=&
		\sum_{k=1}^{n-1}  r(k)k \frac{k}{n-k}\left( \frac{N }{k}-\frac{N}{n}\right)- \frac{N}{n}\sum_{k=1}^{n-1}  r(k)k = 0.
	\end{split}
\end{equation*}
We therefore obtain
\begin{equation*}
	\sum_{N_C=1}^{N-1}\left( P_{CD}(N_C)-P_{DC}(N_C)\right)
	=\left( r(n)\frac{N-n}{n}+1-N\right)c.  
\end{equation*}
This completes the proof of Proposition~\ref{pro1}.

\section{\label{suppC}Proof of Proposition \ref{pro2}}
Based on the result obtained in Appendix~\ref{suppB}, one can obtain
\begin{equation*}
	\sum_{N_C=1}^{N-1}\left( P_{CD}(N_C)-P_{DC}(N_C)\right)
	=\left( r(n,0)\frac{N-n}{n}+1-N\right)c.
\end{equation*}
The proofs of Eq.~\eqref{PCD1} are completed.

According to the identities presented in Appendix~\ref{suppA}, we can obtain
\begin{align*}
		&\sum_{N_C=1}^{N-1}\left( P_{CP}(N_C)-P_{PC}(N_C)\right)\\
		=&\sum_{N_C=1}^{N-1}\sum_{n_C=0}^{n-1}\frac{\binom{N_C-1}{n_C}\binom{N-N_C}{n-n_C-1}}{\binom{N-1}{n-1}} \Pi_C(n_C+1,0,n-n_C-1) \\&
		- 
		\sum_{N_C=1}^{N-1}\sum_{n_C=0}^{n-1} \frac{\binom{N_C}{n_C}\binom{N-N_C-1}{n-n_C-1}}{\binom{N-1}{n-1}}\Pi_P(n_C,0,n-n_C)\\
		=&
		\sum_{N_C=1}^{N-1}\sum_{n_C=0}^{n-1}\frac{\binom{N_C-1}{n_C}\binom{N-N_C}{n-n_C-1}}{\binom{N-1}{n-1}}\left(r(n_C+1,n-n_C-1)c\right. \\&\left.
		-c\right) 
		-
		\sum_{N_C=1}^{N-1}\sum_{n_C=0}^{n-1} \frac{\binom{N_C}{n_C}\binom{N-N_C-1}{n-n_C-1}}{\binom{N-1}{n-1}}\left(r(n_C,n-n_C)c -c\right) \\
		=&
		\sum_{n_C=0}^{n-1}\left(r(n_C+1,n-n_C-1)c-c\right)\\& \cdot \sum_{N_C=1}^{N-1}\frac{\binom{N_C-1}{n_C}\binom{N-N_C}{n-n_C-1}}{\binom{N-1}{n-1}} 
		-\sum_{n_C=0}^{n-1}\left(r(n_C,n-n_C)c -c\right)\\& \cdot
		\sum_{N_C=1}^{N-1} \frac{\binom{N_C}{n_C}\binom{N-N_C-1}{n-n_C-1}}{\binom{N-1}{n-1}} \\
		=&
		\sum_{n_C=0}^{n-2}\left(r(n_C+1,n-n_C-1)c-c\right)\frac{N}{n}+(r(n,0)c\\
		&-c)\frac{N-n}{n} 
		-\sum_{n_C=1}^{n-1}\left(r(n_C,n-n_C)c -c\right)\frac{N}{n} \\&-\left(r(0,n)c-c\right)\frac{N-n}{n}\\
		=&
		\sum_{n_C=1}^{n-1}\left(r(n_C,n-n_C)c-c\right)\frac{N}{n}+(r(n,0)c-c)\\&
		\cdot\frac{N-n}{n} 
		-\sum_{n_C=1}^{n-1}\left(r(n_C,n-n_C)c -c\right)\frac{N}{n} \\&-(r(0,n)c-c)\frac{N-n}{n}\\
		=&\frac{N-n}{n}c(r(n,0)-r(0,n)).
\end{align*}
Then the proof of Eq.~\eqref{PCD2} is completed.

Using Eq.~\eqref{PCD1}, we have
\begin{align*}
		&\sum_{N_P=1}^{N-1}\left( P_{PD}(N_P)-P_{DP}(N_P)\right)\\
		=&\sum_{N_P=1}^{N-1}\sum_{n_P=0}^{n-1}\frac{\binom{N_P-1}{n_P}\binom{N-N_P}{n-n_P-1}}{\binom{N-1}{n-1}} \Pi_P(0,n-n_P-1,n_P+1) \\	&
		- 
		\sum_{N_P=1}^{N-1}\sum_{n_P=0}^{n-1} \frac{\binom{N_P}{n_P}\binom{N-N_P-1}{n-n_P-1}}{\binom{N-1}{n-1}}\Pi_D(0,n-n_P,n_P)\\
		=&
		\sum_{N_P=1}^{N-1}\sum_{n_P=0}^{n-1}\frac{\binom{N_P-1}{n_P}\binom{N-N_P}{n-n_P-1}}{\binom{N-1}{n-1}}\left(\frac{r(0,n_P+1)(n_P+1)}{n}c\right. \\&\left.
		-c -\gamma(n-n_P-1)\vphantom{\sum_{N_C=1}^{N-1}}\right) 
		-
		\sum_{N_P=1}^{N-1}\sum_{n_P=0}^{n-1} \frac{\binom{N_P}{n_P}\binom{N-N_P-1}{n-n_P-1}}{\binom{N-1}{n-1}}\\
		&\cdot\left(\frac{r(0,n_P)n_P}{n}c -\beta n_P\right) \\
		=&
		\left( r(0,n)\frac{N-n}{n}+1-N\right)c\\&-\gamma\sum_{N_P=1}^{N-1}\sum_{n_P=0}^{n-1}\frac{\binom{N_P-1}{n_P}\binom{N-N_P}{n-n_P-1}}{\binom{N-1}{n-1}}
		\left(n-n_P-1\right) 
		\\&+\beta
		\sum_{N_P=1}^{N-1}\sum_{n_P=0}^{n-1} \frac{\binom{N_P}{n_P}\binom{N-N_P-1}{n-n_P-1}}{\binom{N-1}{n-1}}n_P \\
		=&
		\left( r(0,n)\frac{N-n}{n}+1-N\right)c
		-\gamma(n-1)\\&
		\cdot\sum_{N_P=1}^{N-1}\sum_{n_P=0}^{n-1}\frac{\binom{N_P-1}{n_P}\binom{N-N_P}{n-n_P-1}}{\binom{N-1}{n-1}}\\& 	
		+\gamma \sum_{N_P=1}^{N-1}\sum_{n_P=0}^{n-1}\frac{\binom{N_P-1}{n_P}\binom{N-N_P}{n-n_P-1}}{\binom{N-1}{n-1}}n_P\\&
		+\beta
		\sum_{N_P=1}^{N-1}\sum_{n_P=0}^{n-1} \frac{\binom{N_P}{n_P}\binom{N-N_P-1}{n-n_P-1}}{\binom{N-1}{n-1}}n_P\\
		=&
		\left( r(0,n)\frac{N-n}{n}+1-N\right)c
		-\gamma(n-1)\left( \frac{N-n}{n}\right. \\&\left.
		+\sum_{n_P=0}^{n-2}\frac{N}{n}\right)
		+\gamma\frac{(N-2)(n-1)}{2}+\beta\frac{N(n-1)}{2} \\
		=&
		c\left( r(0,n)\frac{N-n}{n}+1-N\right)+(\beta-\gamma)\frac{N(n-1)}{2}.
\end{align*}
Then the proof of Eq.~\eqref{PCD3} is completed.

\end{document}